\documentclass[10pt,conference]{IEEEtran}
\usepackage[pdftex]{graphicx}
\usepackage{amsmath}
\usepackage{multirow}

\newcommand{\bfb}{\mbox{\boldmath $b$}}
\newcommand{\bfc}{\mbox{\boldmath $c$}}
\newcommand{\bfd}{\mbox{\boldmath $d$}}
\newcommand{\bfg}{\mbox{\boldmath $g$}}
\newcommand{\bfu}{\mbox{\boldmath $u$}}

\newcommand{\bfx}{\mbox{\boldmath $x$}}
\newcommand{\bfy}{\mbox{\boldmath $y$}}

\newtheorem{example}{\normalfont \textbf{Example}}
\newtheorem{lemma}{\normalfont \textbf{Lemma}}
\newtheorem{theorem}{\normalfont \textbf{Theorem}}

\begin{document}

\title{Construction of $q$-ary Constant Weight Sequences using a Knuth-like Approach}
\author{\IEEEauthorblockN{Elie N. Mambou and Theo G. Swart}
\IEEEauthorblockA{Dept. of Electrical and Electronic Engineering Science, University of Johannesburg\\
P. O. Box 524, Auckland Park, 2006, South Africa\\
Email: \{emambou, tgswart\}@uj.ac.za}}

\maketitle

\begin{abstract}
We present an encoding and decoding scheme for constant weight sequences, that is, given an information sequence, the construction results in a sequence of specific weight within a certain range. The scheme uses a prefix design that is based on Gray codes. Furthermore, by adding redundant symbols we extend the range of weight values for output sequences, which is useful for some applications.
\end{abstract}

\section{Introduction}

Constant weight (CW) sequences have found applications in the fields of computer science, information security and communications. They play an important role in communication system where high security and confidentiality are needed, because of various properties such as low correlations, balanced value distributions and strong linear complexity, amongst others. 

Practically, CW sequences are used in several modern applications including frequency hopping in GSM networks, detection of unidirectional errors and threshold setting in bar-code implementations. However, our target application is the domain of visible light communication (VLC) where CW sequences can be designed to perform light dimming which decreases flickering issues.  

In the coding theory field, a CW code can be viewed as an error detection and correction code such that all codewords in that code have the same Hamming weight. Binary CW codes are also called $m$ of $n$ codes as each codeword has a length $n$ and $m$ instances of $1$s (weight equals $m$). A special case of $m$ of $n$ code is the 1 of $n$ code, it encodes $\log_2n$ bits in a codeword of length $n$. For instance, the 1 of 2 code has inputs $0$ and $1$ and generates codewords $01$ and $10$. A 1 of 4 code generates codewords $0001$, $0010$, $0100$ and $1000$ given inputs $00$, $01$, $10$ and $11$. In this case the Hamming distance is $d=2$ and each sequence has a weight $w=1$.

There are many algorithms to generate binary CW codes. In \cite{etzion2014}, a construction of CW codes from a given code length is presented. The obtained codes are usually referred to as optical orthogonal codes. This scheme has an efficient algorithm for error code correction and performs the encoding and the decoding of CW codes. However, the construction is limited to specific constant dimension codes. In \cite{skachek2014}, a construction of CW codes based on Knuth's balancing approach \cite{knuth1986} is presented. The flexibility on the tail bits is used to generate CW codes including balanced codes. As in Knuth's algorithm, the information about the changes in the word is carried in the prefix.

In \cite{takayasu2009}, a construction for a set of non-binary constant-weight sequences was proposed with finite period from cyclic difference sets, this is based on the generalization of the binary case proposed in \cite{li}. Other existing works on binary and non-binary constant-weight sequences generated from $q$-ary sequences and cyclic difference sets include \cite{takayasu2011}--\cite{zheng2005}.

The rest of this paper is structured as follows: in Section~\ref{sec2}, we present preliminary work for our construction. Section~\ref{sec3} shows the encoding of the $q$-ary CW sequences based on Gray code prefixes. Section~\ref{sec4} presents the encoding of CW sequences with higher weights, then Section~\ref{sec5} describes the decoding method for this algorithm. Finally, in Section~\ref{sec6} we analyze the redundancy for $q$-ary CW sequences.

\section{Preliminaries}\label{sec2}

\subsection{Balancing of $q$-ary information sequences}\label{sec2.1}

Let $\bfx = x_1 x_2 \ldots x_k$ be a $q$-ary information sequence of length $k$, with $x_i \in \{0, 1, \ldots, q-1\}$, and $w(\bfx)$ be the weight of $\bfx$, that is $w(\bfx) = \sum_{i=1}^{k} x_i$. When $w(\bfx) = k(q-1)/2$, $\bfx$ is said to be balanced, and $\beta_{k,q}$ represents this balancing value.

A construction was presented in \cite{swart2009} for the balancing of $q$-ary sequences, that stipulates that for any $q$-ary information sequence of $k$ symbols, there is always a way to achieve balancing by adding (modulo $q$) one sequence from a set of weighting sequences to $\bfx$. The weighting sequence, $\bfb(s,p) = b_1 b_2 \dots b_k$ is defined as follows:
\begin{equation*}
b_i = \begin{cases}
            s,           & i-1 \geq p, \\
            s+1 \pmod q, & i-1 < p,
      \end{cases}
\end{equation*}
with $s$ and $p$ positive integers such that $0 \leq s \leq q-1$ and $0 \leq p \leq k-1$. Let $z$ be the counter of all possible weighting sequences, $z = sk+p$ and $0\leq z \leq kq-1$. We use $\bfb(s,p)$ and $\bfb(z)$ interchangeably to denote the $z$-th weighting sequence.

We define $\bfy = \bfx \oplus_q \bfb(z)$, as the addition (modulo $q$) of the information sequence $\bfx$ to the weighting sequence $\bfb(z)$.

\begin{example}\label{ex1}
Consider $q=3$ and $k=4$, where we want to balance the ternary information sequence of length four, 2102.

The process (illustrated on the next page) leads to the construction of at least one balanced sequence. Bold weights indicate that the desired weight has been obtained. Fig.~\ref{fex1} presents the weight progression $w(\bfy)$ vs. $z$ for this example. The red line represents the balancing value, $\beta_{4,3}=4$.

\begin{equation*}
    \begin{array}{c@{\quad\quad}c@{\quad\quad}c}
        z  & \bfx \oplus_q \bfb(z) = \bfy & w(\bfy)    \\
        \hline
        0  & 2102 \oplus_3 0000 = 2102    & 5          \\
        1  & 2102 \oplus_3 1000 = 0102    & 3          \\
        2  & 2102 \oplus_3 1100 = 0202    & \textbf{4} \\
        3  & 2102 \oplus_3 1110 = 0212    & 5          \\
        4  & 2102 \oplus_3 1111 = 0210    & 3          \\
        5  & 2102 \oplus_3 2111 = 1210    & \textbf{4} \\
        6  & 2102 \oplus_3 2211 = 1010    & 2          \\
        7  & 2102 \oplus_3 2221 = 1020    & 3          \\
        8  & 2102 \oplus_3 2222 = 1021    & \textbf{4} \\
        9  & 2102 \oplus_3 0222 = 2021    & 5          \\
        10 & 2102 \oplus_3 0022 = 2121    & 6          \\
        11 & 2102 \oplus_3 0002 = 2101    & \textbf{4} \\
    \end{array}
\end{equation*}
\begin{figure}[!h]
    \centering
    \includegraphics[width=0.85\linewidth]{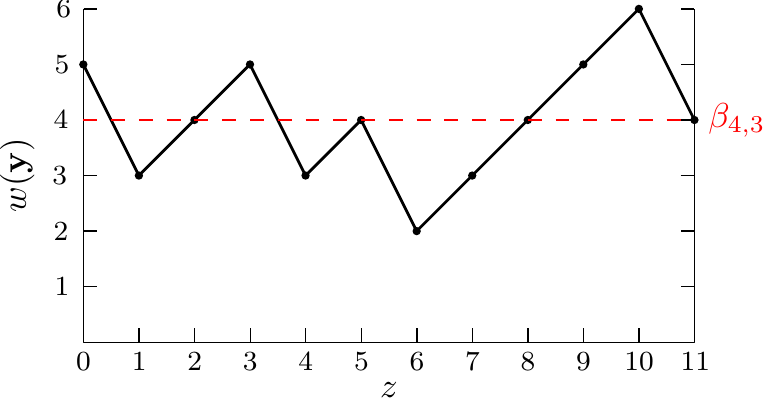}
    \vspace{-6pt}\caption{Weight progression $w(\bfy)$ vs. $z$ for Example~\ref{ex1}.}
    \label{fex1}
\end{figure}

From this we can see that weight values other than $\beta_{4,3}$ are easily achievable, and this will form the basis for our CW algorithm in Section~\ref{sec3}.
\end{example}

\subsection{$q$-ary Gray codes}\label{sec2.2}

Gray codes were invented by Gray \cite{gray1953} for solving problems in pulse code communication, and have been extended to various other applications. 

We define an $(r',q)$-Gray code as the set of $q$-ary Gray sequences of length $r'$. This set presents the property that any two consecutive codewords (provided that the sequences being mapped from are listed in the normal lexicographic order) differ only in one symbol position and the difference between any two consecutive sequences' weights is $\pm 1$.

Let $\bfd = d_1 d_2 \ldots d_{r'}$ be any sequence within the set of $q$-ary sequences of length $r'$, listed in the normal lexicographic order. These sequences are mapped to $(r', q)$-Gray code sequences, $\bfg = g_1 g_2 \ldots g_{r'}$, by making use of the following encoding and decoding algorithms \cite{guan1998} for Gray codes.

\textbf{Encoding algorithm for \boldmath$(r', q)$-Gray code:} Let $S_i$ be the sum of the first $i-1$ symbols of $\bfg$, with $2 \leq i \leq r'$ and $g_1 = d_1$. Then
\begin{equation*}
    S_i=\sum_{j=1}^{i-1}g_j, \quad\text{and}\quad
    g_i = \begin{cases}
                d_i,     & \text{if } S_i \text{ is even}, \\
                q-1-d_i, & \text{if } S_i \text{ is odd}.
          \end{cases}
\end{equation*}
The parity of $S_i$ determines symbols of $\bfg$ from $\bfd$. If $S_i$ is even then the symbol stays the same, otherwise the $q$-ary complement of the symbol is taken.

\textbf{Decoding algorithm for \boldmath$(r', q)$-Gray code:} As before, $S_i$ is the sum of the first $i-1$ symbols of $\bfg$, with $2 \leq i \leq r'$ and $d_1 = g_1$. Then
\begin{equation*}
    S_i=\sum_{j=1}^{i-1}g_j, \quad\text{and}\quad
    d_i = \begin{cases}
                g_i,     & \text{if } S_i \text{ is even}, \\
                q-1-g_i, & \text{if } S_i \text{ is odd}.
          \end{cases}
\end{equation*}

\subsection{Encoding of balanced $q$-ary information sequences based on Gray code prefixes}\label{sec2.3}

A method for the encoding and decoding of balanced sequences based on Gray code prefixes was presented in \cite{mambou2016}. It was proven that any $q$-ary information sequence of length $k$ can be balanced by adding (modulo $q$) an appropriate weighting sequence $\bfb(z)$, and prefixing a redundant symbol $u$ with a Gray code sequence. This method consists of generating $\bfy = \bfx \oplus \bfb(z)$ outputs as presented in Section~\ref{sec2.1}. However, $q$-ary Gray code prefixes of length $r'=\log_qk+1$, are appended to $\bfy$, this only works for information sequences where $k=q^t$, with $t$ being a positive integer. For each output sequence, the $q$-ary representation of the index $z$ is converted into its corresponding Gray code representation, as described in Section~\ref{sec2.2}. Then an extra digit $u$, is appended to the Gray code prefix to enforce balancing as necessary, where $u = \beta_{n,q} - w(\bfc')$, $n=k+r'+1$ being the overall length of output sequence $\bfc = [u|\bfg|\bfy]$ and $\bfc' = [\bfg|\bfy]$, where $|$ represents concatenation of sequences.

\begin{example}\label{ex23}
We consider encoding the ternary information sequence of length three, 102.

The process below shows the encoding, where the underlined part represents the prefix and the bold symbol is the symbol $u$. The Gray code length is $r'=2$, the total length of the transmitted sequence is $n=6$ and $\beta_{6,3}=6$.
\begin{equation*}
    \begin{array}{c@{\quad\quad}ccc}
        z & \bfx \oplus_q \bfb(z) = \bfy & \bfc = [u|\bfg|\bfy]        & w(\bfc)    \\
        \hline
        0 & 102 \oplus_3 000 = 102       & \underline{\textbf{0}00}102 & 3          \\
        1 & 102 \oplus_3 100 = 202       & \underline{\textbf{1}01}202 & \textbf{6} \\
        2 & 102 \oplus_3 110 = 212       & \underline{\textbf{0}02}212 & 7          \\
        3 & 102 \oplus_3 111 = 210       & \underline{\textbf{0}12}210 & \textbf{6} \\
        4 & 102 \oplus_3 211 = 010       & \underline{\textbf{0}11}010 & 3          \\
        5 & 102 \oplus_3 221 = 020       & \underline{\textbf{0}10}020 & 3          \\
        6 & 102 \oplus_3 222 = 021       & \underline{\textbf{1}20}021 & \textbf{6} \\
        7 & 102 \oplus_3 022 = 121       & \underline{\textbf{0}21}121 & 7          \\
        8 & 102 \oplus_3 002 = 101       & \underline{\textbf{0}22}101 & \textbf{6} \\
    \end{array}
\end{equation*}
We have four occurrences of balanced outputs: $\underline{\textbf{1}01}202$, $\underline{\textbf{0}12}210$, $\underline{\textbf{1}20}021$ and $\underline{\textbf{0}22}101$. The encoding of $(2,3)$-Gray code prefixes can be followed from Table~\ref{tab:4}. Fig.~\ref{fexp23} presents the weight progression for this example.
\begin{figure}[!h]
    \centering
    \includegraphics[width=0.7\linewidth]{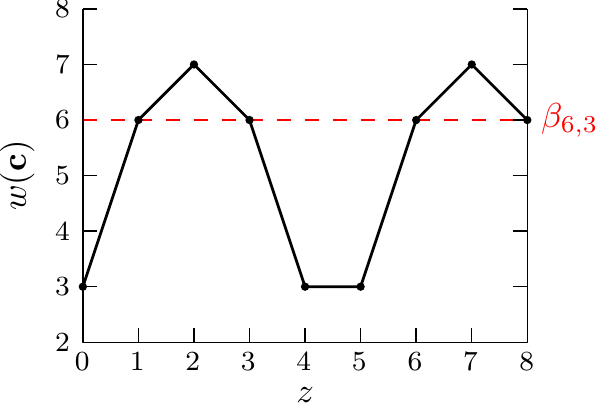}
    \vspace{-6pt}\caption{Weight progression $w(\bfc)$ vs. $z$ for Example~\ref{ex23}.}
    \label{fexp23}
\end{figure}
\begin{table}[h]
    \centering
    \caption{Encoding of $(2,3)$-Gray code}\label{tab:4}\vspace{-10pt}
    \begin{tabular}{ccccc}
        \hline\hline
        $z$ & $s,p$ & $\bfb(z)$ & Sequence ($\bfd$) & Gray code ($\bfg$) \\
        \hline
        0   & $0,0$ & $000$     & $00$              & $00$               \\
        1   & $0,1$ & $100$     & $01$              & $01$               \\
        2   & $0,2$ & $110$     & $02$              & $02$               \\
        3   & $1,0$ & $111$     & $10$              & $12$               \\
        4   & $1,1$ & $211$     & $11$              & $11$               \\
        5   & $1,2$ & $221$     & $12$              & $10$               \\
        6   & $2,0$ & $222$     & $20$              & $20$               \\
        7   & $2,1$ & $022$     & $21$              & $21$               \\
        8   & $2,2$ & $002$     & $22$              & $22$               \\
        \hline\hline
    \end{tabular}
\end{table}
\end{example}

The proposed algorithm for the encoding of $q$-ary CW sequences makes use of this process.

\section{Encoding of $q$-ary CW sequences based on Gray code prefixes}\label{sec3}

We now extend this work on encoding balanced sequences, to encode CW sequences. Let an $(n,k,W,q)$ CW sequence be the $q$-ary CW sequence of length $n$, weight $W$ with $k$ information symbols.

The encoding process consists of using the $(r',q)$-Gray code to encode the index $z$ as prefix. The Gray code sequence, $\bfg=g_1g_2\dots g_r'$ is prefixed to all $kq$ outputs as $\bfc'=[\bfg|\bfy]$. We have to find the appropriate length of $(r',q)$-Gray code prefixes that can uniquely match the $kq$ weighting sequences.

We impose a condition that $k=q^t$, where $t$ is a positive integer, therefore
\begin{equation}\label{eq1}
    r'=\log_q(kq)=\log_q(k)+1,
\end{equation}
such that the cardinalities of the Gray code set and that of the weighting sequences are equal. As in the case for balanced $q$-ary sequences, an extra digit $u$ is also added to $\bfc'$ to control the overall weight of the CW sequence.
 
For the construction of $(n,k,W,q)$ CW sequences, let $\bfc = [u|\bfg|\bfy]$ be the concatenation of $u$, $\bfg$ and $\bfy$. For a specific $z$, if $W \geq w(\bfc')$, provided that $u \in \{0,1,\dots,q-1\}$, $u = W-w(\bfc')$, else if $W < w(\bfc')$ then $u=0$. The overall encoded sequence has a length $n$ and the prefix has a length of $r = r'+1=\log_qk+2$.

\textbf{Encoding algorithm for $q$-ary CW sequence:}
\begin{enumerate}
    \item The length of the $(r',q)$-Gray code prefix is calculated as in \eqref{eq1}, $r'=\log_qk +1$.
    \item Incrementing through $z$, determine weighting sequences, $\bfb(z)$ and add them to $\bfx$, $\bfy=\bfx\oplus \bfb(z)$.
    \item For each increment of $z$, determine the corresponding Gray code sequence, $\bfg$, using the Gray code encoding algorithm presented in Section~\ref{sec2.2}, append it to sequence $\bfy$ and obtain sequence $\bfc'$.
    \item Finally, append the extra digit, $u$, to sequence $\bfc'$, where $u=W-w(\bfc')$ provided that $u\in \{0,1,\dots,q-1\}$, otherwise $u=0$.
\end{enumerate}
\begin{lemma}\label{lm1}
For any $q$-ary information sequence $\bfx$ of length $k$, where parameters $k$ and $q$ are not coprime, we can find a $\bfb(z)$ such that the weight of $\bfy = \bfx \oplus_q \bfb(z)$ is $\omega_1\leq w(\bfy) \leq \omega_2$, where $\omega_1 = \beta_{k,q}-(q-1) = \frac{(k-2)(q-1)}{2}$ and $\omega_2 = \beta_{k,q}+(q-1) = \frac{(k+2)(q-1)}{2}$.
\end{lemma}
\begin{proof}
It was proven in \cite{swart2009} that the weight progression graph of $\bfy=\bfx\oplus_q \bfb(z)$ represents a path with increases of 1 and decreases of $q-1$, and it is such that min$\{w(\bfy)\}$ $\leq \beta_{k,q}$ and max$\{w(\bfy)\}$ $\ge\beta_{k,q}$. The constraint on the information sequences $k=q^t$, implies that $k$ and $q$ are not coprime. Lemma 1 is true if min$\{w(\bfy)\}\ge t_1$ and max$\{w(\bfy)\}\leq t_2$. The weighting sequence $\bfb(z)$ is such that $0\leq w(\bfb(z))\leq k(q-1)$, so the range of weighting sequence is greater than the range of $[\omega_2-\omega_1]=2(q-1)$. Therefore $\min\{w(\bfy)\}\geq \omega_1$ and $\max\{w(\bfy)\} \leq \omega_2$.
\end{proof}

\begin{theorem}\label{th01}
An $(n,k,W,q)$ CW sequence can be constructed from any $q$-ary information sequence $\bfx$ of length $k$ where
\begin{equation}\label{eq3}
    \frac{(k-2)(q-1)}{2} \leq W \leq \frac{(k+2r'+4)(q-1)}{2}.
\end{equation}
\end{theorem}\vspace{-12pt}
\begin{proof}
According to Lemma~\ref{lm1}, $\frac{(k-2)(q-1)}{2}\leq W\leq \frac{(k+2)(q-1)}{2}$, however because of the flexibility of $u$ and $\bfg$, the upper bound increases to $\frac{(k+2r'+4)(q-1)}{2}$.
\end{proof}

\begin{example}\label{ex2}
We want to encode the ternary sequence $\bfx=212$ into a CW sequence of weight $W=8$. The condition $k = q^t \Rightarrow 3$, is fulfilled and the Gray code length is $r'=\log_33+1=2$, according to \eqref{eq1}.

The cardinality of the $(2,3)$-Gray code equals that of the weighting sequences, with $kq=9$. The overall length of the CW sequence is $n=3+2+1=6$. The weight range as presented in \eqref{eq3} is such that $2\leq W \leq 10$, therefore it is possible to construct $(6,3,8,3)$ CW sequences. The process below presents the encoding, with bold weight values indicating that the desired $W=8$ is attained, the underline part represents the prefix and the bold symbol is the digit $u$.
\begin{equation*}
    \begin{array}{cccc}
        z & \bfx \oplus_q \bfb(z) = \bfy & \bfc = [u|\bfg|\bfy]& w(\bfc)\\\hline
        0 & 212 \oplus_3 000 = 212       & \underline{\textbf{0}00}212 & 5          \\
        1 & 212 \oplus_3 100 = 012       & \underline{\textbf{0}01}012 & 4          \\
        2 & 212 \oplus_3 110 = 022       & \underline{\textbf{2}02}022 & \textbf{8} \\
        3 & 212 \oplus_3 111 = 020       & \underline{\textbf{0}12}020 & 5          \\
        4 & 212 \oplus_3 211 = 120       & \underline{\textbf{0}11}120 & 5          \\
        5 & 212 \oplus_3 221 = 100       & \underline{\textbf{0}10}100 & 2          \\
        6 & 212 \oplus_3 222 = 101       & \underline{\textbf{0}20}101 & 4          \\
        7 & 212 \oplus_3 022 = 201       & \underline{\textbf{2}21}201 & \textbf{8} \\
        8 & 212 \oplus_3 002 = 211       & \underline{\textbf{0}22}211 & \textbf{8} \\
    \end{array}
\end{equation*}
In this case, there are three occurrences of a $(6,3,8,3)$ CW sequence: $\underline{\textbf{2}02}022$, $\underline{\textbf{2}21}201$ and $\underline{\textbf{0}22}211$. Fig.~\ref{fexp2} presents the weight progression for this example.
\end{example}
\begin{figure}[h]
    \centering
    \includegraphics[width=0.75\linewidth]{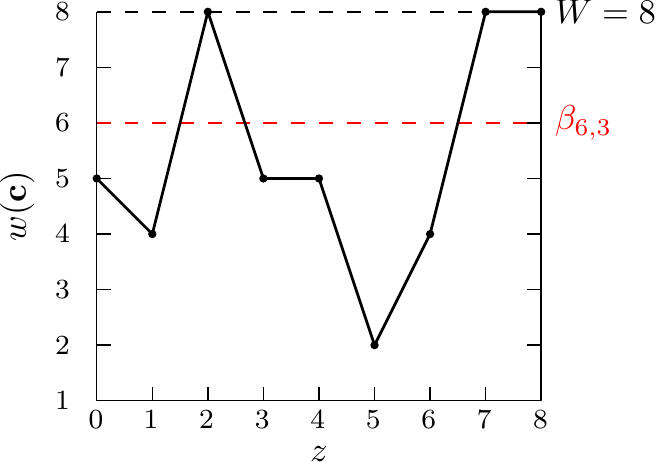}
    \vspace{-6pt}\caption{Weight progression of $(6,3,8,3)$ CW sequence of Example~\ref{ex2}.}
    \label{fexp2}
\end{figure}

\section{Construction of $q$-ary constant weight sequences with extended weight range}\label{sec4}

In the previous section, the construction of $q$-ary CW sequences was achieved with a weight range shown in \eqref{eq3}. However, because of the limited interval, we will present an approach to extend this range.

The method consists of appending a redundant vector $\bfu$ of length $e$ to $\bfc' = [\bfg|\bfy]$, then the output sequence becomes $\bfc = [\bfu|\bfg|\bfy]$. This leads to $(n,k,W,q)$ CW sequences where $n = k+r'+e$.

In general, generating $(n,k,W,q)$ CW sequence from a $q$-ary information sequence of length $k$ with the redundant vector $\bfu$ of length $e$ will lead to an increase of weight range. Combining \eqref{eq3} and $w(\bfu) \in [0, e(q-1)]$ results in a CW sequence $\bfc = [\bfu|\bfg|\bfy]$ of weight $W$ is such that 
\begin{equation}\label{eq5}
\frac{(k-2)(q-1)}{2}\leq W \leq \frac{(k+2r'+2e+2)(q-1)}{2}.
\end{equation}
Example~\ref{ex2} can then be viewed as a special case with $e=1$.

\begin{theorem}\label{th1}
Any $q$-ary information sequence of length $k$ can generate an $(n,k,W,q)$ CW sequence where $W \in \left[\frac{(k-2)(q-1)}{2}, \frac{(k+2r'+2e+1)(q-1)}{2}\right]$.
\end{theorem}
Theorem~\ref{th1} can be proved using a similar argument as in Theorem~\ref{th01}. However, the condition $k=q^t$, where $t$ is a positive integer must still be fulfilled in order to observe equality between the cardinality of $(r',q)$-Gray code prefixes and that of the weighting sequences.

The redundant vector $\bfu = u_1 u_2 \ldots u_e$ is such that $u_i \in \{0, 1, \ldots, q-1\}$ and $w(\bfu) = W - w(\bfc')$ if and only if $W \geq w(\bfc')$, otherwise $\bfu = \textbf{0}$.
\begin{table}[!b]
    \centering
    \caption{Parameters evaluation}\label{tab:3}\vspace{-10pt}
    \begin{tabular}{ccccccc}
        \hline\hline
              & $t$ & $k=q^t$ & $W$         & $n$ & $r'$ & $e$ \\
        \hline
              & 2   & 4       & $[4,9]$     & 10  & 3    & 3   \\
        $q=2$ & 3   & 8       & $[7,13]$    & 16  & 4    & 4   \\
              & 4   & 16      & $[11,17]$   & 24  & 5    & 3   \\
        \hline
              & 1   & 3       & $[5,12]$    & 7   & 2    & 2   \\
        $q=3$ & 2   & 9       & $[13,22]$   & 15  & 3    & 3   \\
              & 3   & 27      & $[32,42]$   & 34  & 4    & 3   \\
        \hline
              & 1   & 4       & $[12,20]$   & 8   & 2    & 2   \\
        $q=4$ & 2   & 16      & $[30,43]$   & 22  & 3    & 3   \\
              & 3   & 64      & $[105,121]$ & 72  & 4    & 4   \\
        \hline\hline
    \end{tabular}
\end{table}
Table~\ref{tab:3} presents some parameters evaluation, that is values of total length $n$ and the achievable range of weights $W$, given an alphabet size $q$ and an information sequence length $k$.
\begin{example}\label{ex4}
Consider the same ternary information sequence $\bfx = 212$ of length 3 as in Example~\ref{ex2}. We would like to generate a $(7,3,12,3)$ CW sequence of weight $W=12$ and $n=7$ as described in Table~\ref{tab:3}.

We observe from \eqref{eq3} that the information sequence 212 can only generate $(n,k,W,q)$ CW sequences with $W \in [2,10]$. In order to extend this range of weight, we append a ternary redundant vector $\bfu$ of length $e$ to $\bfc'$. For $e=2$, \eqref{eq5} stipulates that $W \in [2,12]$, therefore the weight $W=12$ can be obtained.
\begin{itemize}
    \item There are nine possible vectors, $\bfu$ of length two that can be appended to $\bfc'$ as follow: 00, 01, 02, 10, 11, 12, 20, 21, 22.
    \item The redundant vector weight is such that $w(\bfu) = 12-w(\bfc')$ for $W \geq w(\bfc')$.
    \item The length of Gray code prefixes is $r' = \log_33+1=2$.
    \item Below we repeat the process defined in Section~\ref{sec3}:
\end{itemize}
\begin{equation*}
	\begin{array}{cccc}
        z & \bfx \oplus_q \bfb(z) = \bfy & \bfc = [\bfu|\bfg|\bfy]      & w(\bfc)     \\
        \hline
        0 & 212 \oplus_3 000 = 212       & \underline{\textbf{00}00}212 & 5           \\
        1 & 212 \oplus_3 100 = 012       & \underline{\textbf{00}01}012 & 4           \\
        2 & 212 \oplus_3 110 = 022       & \underline{\textbf{00}02}022 & 6           \\
        3 & 212 \oplus_3 111 = 020       & \underline{\textbf{00}12}020 & 5           \\
        4 & 212 \oplus_3 211 = 120       & \underline{\textbf{00}11}120 & 5           \\
        5 & 212 \oplus_3 221 = 100       & \underline{\textbf{00}10}100 & 2           \\
        6 & 212 \oplus_3 222 = 101       & \underline{\textbf{00}20}101 & 4           \\
        7 & 212 \oplus_3 022 = 201       & \underline{\textbf{00}21}201 & 6           \\
        8 & 212 \oplus_3 002 = 211       & \underline{\textbf{22}22}211 & \textbf{12} \\
	\end{array}
\end{equation*}	
The underline sequence represents the overall prefix where the bold part is the redundant vector $\bfu$ and the rest is the Gray code prefix. We observe that by adding a redundant vector of length $e=2$ in the process, there is one occurrence of a $(7,3,12,3)$ CW sequence which is $\underline{\textbf{22}22}211$. Fig.~\ref{fexp23p} presents the weight progression for this example.
\end{example}
\begin{figure}[!h]
	\centering
	\includegraphics[width=0.7\linewidth]{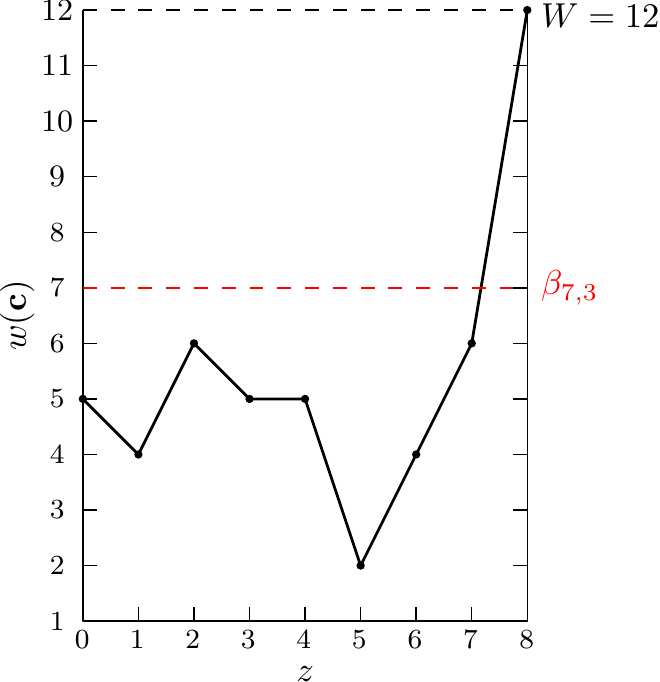}
	\vspace{-6pt}\caption{Weight progression for Example~\ref{ex4}.}
	\label{fexp23p}
\end{figure}	

We observe that by adding a redundant vector of length $e=2$, we increased the weight range from $[2,10]$ to $[2,12]$, that is the weight range increases proportionally with the length of the redundant symbols.

\section{Decoding of $q$-ary constant weight sequences}\label{sec5}

The decoding of $(n,k,W,q)$ CW sequences follows the same process as the decoding of balanced sequences presented in \cite{mambou2016}. This process consists of the following steps:
\begin{enumerate}
    \item The redundant vector $\bfu$ is dropped, then the $r'$ symbols are extracted as the Gray code prefix and converted to $z$ as presented in Section~\ref{sec2.2}.
    \item $z$ is used to determine the parameters $s$ and $p$, then $\bfb(s,p)$ can be derived.
    \item Finally, the original sequence is recovered through $\bfx= \bfy \ominus_q \bfb(s,p)$.
\end{enumerate}

\begin{example}\label{ex3}
We want to decode the $(7,4,14,4)$ CW sequence, $\underline{\textbf{2}31}3113$. 	
	
The illustration of the decoding process for the $(2,4)$-Gray code prefixes is presented in Table \ref{tab:5}, the bold row represents the decoding of the extracted Gray code prefix, $31$.
\begin{itemize}
	\item The redundant vector $\bfu=2$ is dropped. Then the Gray code sequence of length 2, is extracted as $31$.		
	\item The Gray code $\bfg=31$ corresponds to $\bfd = 32$, and index $z=14$ according to Table~\ref{tab:5}. This implies that $s=3$ and $p=2$, therefore $\bfb(3,2) = 0033$.
	\item Finally, the information sequence is recovered as
    \begin{equation*}
        \bfx = \bfy \ominus_q \bfb(s,p) = 3113 \ominus_3 0033 = 3120.
    \end{equation*}
\end{itemize}	
\begin{table}[!b]
	\centering
	\caption{Decoding of $(2,4)$-Gray code}\label{tab:5}\vspace{-10pt}
	\begin{tabular}{ccccc}
		\hline\hline
		Gray code ($\bfg$) & Sequence ($\bfd$) & $z$         & $s,p$          & $\bfb(s,p)$     \\
		\hline
		00                 & 00                & 0           & $0,0$          & 0000          \\
		01                 & 01                & 1           & $0,1$          & 1000          \\
		02                 & 02                & 2           & $0,2$          & 1100          \\
		03                 & 03                & 3           & $0,3$          & 1110          \\
		13                 & 10                & 4           & $1,0$          & 1111          \\
		12                 & 11                & 5           & $1,1$          & 2111          \\
		11                 & 12                & 6           & $1,2$          & 2211          \\
		10                 & 13                & 7           & $1,3$          & 2221          \\
		20                 & 20                & 8           & $2,0$          & 2222          \\
		21                 & 21                & 9           & $2,1$          & 3222          \\
		22                 & 22                & 10          & $2,2$          & 3322          \\
		23                 & 23                & 11          & $2,3$          & 3332          \\
		33                 & 30                & 12          & $3,0$          & 3333          \\
		32                 & 31                & 13          & $3,1$          & 0333          \\
		\textbf{31}        &\textbf{32}        & \textbf{14} & $\mathbf{3,2}$ & \textbf{0033} \\
		30                 & 33                & 15          & $3,3$          & 0003          \\
		\hline\hline
	\end{tabular}
\end{table}	
\end{example}
\begin{table}[!b]
	\centering
	\caption{Comparison of our construction against the full set of CW and balanced sequences}\label{tab:6}\vspace{-10pt}
	\begin{tabular}{cc|ccccc}
        \hline\hline
        \multicolumn{2}{c|}{$W$}                 & $q$ & $n$ & $k$ & $\mathcal{N}_1$ & $\mathcal{N}_2$ \\
        \hline
        \multirow{5}{*}{$\beta_{n,q}-q+1$} & 3   & 2   & 7   & 4   & 35              & 16              \\
                                           & 5   & 2   & 12  & 8   & 792             & 256             \\
                                           & 10  & 2   & 21  & 16  & 352716          & 65536           \\ 
                                           & 3   & 3   & 5   & 3   & 30              & 27              \\
                                           & 10  & 3   & 12  & 9   & 58278           & 19683           \\
                                           & 6   & 4   & 6   & 4   & 336             & 256             \\
		\hline
        \multirow{5}{*}{$\beta_{n,q}$}     & 4   & 2   & 7   & 4   & 35              & 16              \\
                                           & 6   & 2   & 12  & 8   & 924             & 256             \\ 
                                           & 11  & 2   & 21  & 16  & 352716          & 65536           \\
                                           & 5   & 3   & 5   & 3   & 51              & 27              \\
                                           & 12  & 3   & 12  & 9   & 737789          & 19683           \\
                                           & 9   & 4   & 4   & 6   & 580             & 256             \\
        \hline
        \multirow{5}{*}{$\beta_{n,q}+q$}   & 6   & 2   & 8   & 4   & 28              & 16              \\
                                           & 9   & 2   & 13  & 8   & 715             & 256             \\
                                           & 13  & 2   & 22  & 16  & 497420          & 65536           \\
                                           & 9   & 3   & 6   & 3   & 50              & 27              \\
                                           & 16  & 3   & 13  & 9   & 129844          & 19683           \\
                                           & 15  & 4   & 7   & 4   & 728             & 256             \\
		\hline\hline
	\end{tabular}
\end{table}

\section{Redundancy and Complexity Analysis}\label{sec6}

The redundancy of the presented method comes from the Gray code prefix and the redundant vector. The overall redundancy is $r=\log_qk+e+1 \Rightarrow k=q^{r-1-e}$. Therefore, the total length of encoded CW sequence is $n=k+\log_qk+e+1$. An encoding and decoding algorithm for balanced $q$-ary sequences based on Gray code prefixes was presented in \cite{mambou2016}, with a redundancy of $r=\log_qk+1$, which was compared against existing constructions. However, the two above redundancies differ only with the parameter $e$ which gets larger for high values of $W$ and equals zero for $W=\beta_{n,q}$.

Table~\ref{tab:6} presents the comparison of cardinalities for the full set of CW versus the ones for information sequences. Here $\mathcal{N}_1$ represents the cardinality of $q$-ary CW sequences for specific $W$ of length $n$ while $\mathcal{N}_2$ represents the cardinality of $q$-ary information sequences of length $k$. To have a $q$-ary CW sequence of weight $W$ and length $k$, one clearly requires enough parity bits $r$ such that $\mathcal{N}_1 \ge \mathcal{N}_2=q^k$, where $n=k+r$.

Our method requires $\mathcal{O}(qk\log_qk)$ digit operations for the encoding and $\mathcal{O}(k)$ digit operations for the decoding process; this complexity is similar to the one in construction \cite{mambou2016}.

\section{Conclusion}\label{sec:conclusion}

An efficient algorithm was proposed for encoding and decoding $(n,k,W,q)$ CW sequences based on Gray code prefixes, with a simple method to extend the achievable weight range, if necessary. The construction does not make use of memory-consuming lookup tables, and only simple operations such as addition and subtraction are needed, and most of the decoding process can be performed in parallel. Seeing as the proposed method is only applicable to information sequences of length $k$ where $k=q^t$, the most obvious improvement would be to extend this algorithm to the case where $k \neq q^t$.


\begin{thebibliography}{11}

\bibitem{etzion2014}
    T. Etzion and A. Vardy,
    ``A new construction for constant weight codes'',
    \emph{Int. Symp. Inform. Theory Applic.},
    Melbourne, Australia,
    Oct. 26--29, 2014,
    pp. 338--342. 

\bibitem{skachek2014}
    V. Skachek and K. A. S. Immink,
    ``Constant weight codes: an approach based on Knuth's balancing method,''
    \emph{IEEE J. Sel. Areas Commun.},
    vol. 32, no. 5,
    pp. 908--918,
    Apr. 2014.

\bibitem{knuth1986}
    D. E. Knuth,
    ``Efficient balanced codes,''
    \emph{IEEE Trans. Inform. Theory},
    vol. 32, no. 1,
    pp. 51--53,
    Jan. 1986.

\bibitem{takayasu2009}
    T. Kaida and J. Zheng,
    ``A note on constant weight sequences over $q$-ary from cyclic difference sets,''
    in \emph{Proc. Int. Workshop Signal Design Applic. Commun.},
    Fukuoka City, Japan,
    Oct. 19--23, 2009,
    pp. 112--113.

\bibitem{li}
    N. Li, X. Zheng and L. Hu,
    ``Binary constant weight codes based on cyclic difference sets,''
    in \emph{IEICE Trans. Fundamentals Electron. Commun. Comput. Sci.},
    vol. E91-A, no. 5,
    pp. 1288--1292,
    Jan. 2010.

\bibitem{takayasu2011}
    T. Kaida and J. Zheng,
    ``On $q$-ary constant weight sequences from cyclic difference sets,''
    in \emph{Proc. Int. Workshop Signal Design Applic. Commun.},
    Guilin, China,
    Oct. 10--14, 2011,
    pp. 126--129.

\bibitem{chee2007}
    Y. M. Chee and S. Ling,
    ``Constructions for $q$-ary constant-weight codes,''
    \emph{IEEE Trans. Inform. Theory},
    vol. 53, no. 1,
    pp. 135--146,
    Jan. 2007.

\bibitem{zheng2005}
    F.~Zheng,
    ``Properties of $m$-sequence and construction of constant weight codes,''
    \emph{IEICE Trans. Fundamentals Electron. Commun. Comput. Sci.},
    vol. E88-A, no. 12,
    pp. 3675--3676,
    Dec. 2005.

\bibitem{swart2009}
    T. G. Swart and J. H. Weber,
    ``Efficient balancing of $q$-ary sequences with parallel decoding,''
    in \emph{Proc. IEEE Int. Symp. Inform. Theory},
    Seoul, Korea,
    Jun. 28--Jul. 3, 2009,
    pp. 1564--1568.

\bibitem{gray1953}
    F. Gray,
    ``Pulse code communication,''
    \emph{U. S. Patent 2632058},
    Mar. 1953.

\bibitem{guan1998}
    D.-J. Guan,
    ``Generalized Gray codes with applications,''
    in \emph{Proc. Nat. Sci. Council, Rep. of China, Part A},
    vol. 22, no. 6,
    pp. 841--848,
    1998.

\bibitem{mambou2016}
    E. N. Mambou and T. G. Swart,
    ``Encoding and decoding of balanced $q$-ary sequences using a Gray code prefix,''
    in \emph{Proc. IEEE Int. Symp. Inform. Theory},
    Barcelona, Spain,
    Jul. 10--15, 2016,
    pp. 380--384.

\end{thebibliography}
\end{document}